\documentclass[conference]{IEEEtran}     
     
\IEEEoverridecommandlockouts                              
\usepackage{indentfirst}
\usepackage{graphics}
\usepackage{graphicx}
\usepackage{setspace}
\usepackage{url}
\usepackage{amsmath}

\usepackage{amsmath}
\usepackage{amsmath,amsthm}
\usepackage{amsfonts}
\usepackage{graphicx}
\graphicspath{{/texfiles/}}

\newtheorem{defn}{Definition}[section]

\newtheorem{theorem}{Theorem}[section]

\newtheorem{exm}{Example}[section]

\title{
Market-Based Power Allocation for a Differentially Priced FDMA System
}

\author{\IEEEauthorblockN{Mohammad Hassan Lotfi}
\IEEEauthorblockA{ESE Department\\
University of Pennsylvania\\
Philadelphia, PA, USA, 19104\\
Email: lotfm@seas.upenn.edu}
\and
\IEEEauthorblockN{George Kesidis}
\IEEEauthorblockA{EE and CS\&E Departments\\
Pennsylvania State University\\
University Park, PA, USA, 16803\\
Email: kesidis@engr.psu.edu}
\and
\IEEEauthorblockN{Saswati Sarkar}
\IEEEauthorblockA{ESE Department\\
University of Pennsylvania\\
Philadelphia, PA, USA, 19104\\
Email: swati@seas.upenn.edu}
\thanks{This work has been supported by  NSF grants CNS 1116626, 
CNS 1115547, CNS 1217730.}}

\begin{document}

\maketitle
\thispagestyle{empty}
\pagestyle{empty}

\begin{abstract}
In this paper, we study the problem of differential pricing and QoS assignment by a broadband data provider. In our model, the broadband data provider decides on the power allocated to an end-user not only based on parameters of the transmission medium, but also based on the price the user is willing to pay. In addition, end-users bid the price that they are willing to pay to the Base Station (BS) based on their channel condition, the throughput they require, and their belief about other users' parameters. We will characterize the optimum power allocation by the BS which turns out to be a modification of the solution to the well-known water-filling problem. We also characterize the optimum bidding strategy of end-users using the belief of each user about the cell condition.
\end{abstract}

\section{Introduction}\label{section:introduction}

The exponential growth of the demand for broadband data in recent years and the high cost of network expansion led broadband data providers to revise their policies and prices imposed on content providers and end-users \cite{SoumyaSurvey}, \cite{nonneutrality}. An evidence of such trend is the increasing incentive of broadband data providers to change their flat-rate pricing scheme with a more complex one to control the congestion on the network and to increase their profit. 



In this paper, we study the problem of differential pricing and QoS assignment by a broadband data provider. In our model, the broadband provider decides on the resources allocated to an end-user, and subsequently the QoS seen by that user, not only based on parameters of the transmission medium, but also based on the price the user is willing to pay. An FDMA system with a single Base Station (BS) is considered. The BS has a limited downlink power budget, and is willing to allocate the available power such that to maximize the gross profit. An end-user bids the price that she is willing to pay to the BS based on her own channel condition, the throughput she requires, and her belief about other users' parameters. Subsequently, the BS decides on the amount of power she allocates to end-users based on the price they bid, their channel quality, and the throughput needed by users. 

Although the bidding scheme is not convenient to use by an end-user in everyday data usage, this scheme can be used by users with a bad channel quality and high valuation for data to increase their throughput through appropriately adjusting their price. In addition, this gives an end-user with a good channel quality the ability to use data for a cheaper price. Another application for this scheme is primary/secondary markets \cite{Gaurav}.    

In this paper, we will characterize the optimum power allocation by the BS. We will show that the solution to this problem is a modification of the solution to the well-known water-filling \cite{Boyd} problem. We also characterize the optimum bidding strategy of end-users using the belief of each user about the cell condition. Furthermore, we will investigate the relation between the communication parameters of a user and the price she pays and subsequently the allocated power to her using simulations.  The social welfare of end-users is also compared in two cases: (i) the mentioned differential QoS and pricing scheme and (ii) the flat-rate pricing scheme.



The problem of resource allocation through pricing in wired and wireless networks has received a lot of attention \cite{uplink} - \cite{downlink2}. In most of the works about pricing-based power control in wireless networks, such as  \cite{uplink}-\cite{uplink2}, the problem of power control in \textit{uplink} of a cellular network is considered. However,  \cite{downlink} and \cite{downlink2} investigated the problem of downlink power allocation in a CDMA system. The closest work to ours is \cite{downlink2}. In \cite{downlink2}, authors investigated a bidding scheme in which users bid for a time-frame, and the BS allocates the time-frame based on prices and channel qualities. Authors focused on the macro-level view of the interaction between the BS and end-users, and do not involve wireless intricacies in the model. The decisions of users are considered to be made independent of each other and based on a demand function. They mainly focus on the optimal allocation and revenue given the price vector, and do not characterize the optimum bidding strategies of users. By characterizing the demand function, we capture the nature of end-user's optimization, and characterize both the optimum allocation and bidding strategies. 

Another distinction of our work with previous works is to consider more strategic users that only pay for the amount of data they need. In other words, users have a desired throughput that needs to be met, and they only pay for the data up to their desired throughput and not more. 

The paper is organized as follows. In Section \ref{section:SystemModel}, we describe the model and formulate a two stage optimization problem. In Section \ref{section:optimal}, we solve the optimization problems and characterize the optimum power allocation and bidding strategies. Simulation results are presented in Section \ref{section:simulation}. We conclude the paper and present future works in Section \ref{section:conclusion}.

\section{System Model and Problem Formulation}\label{section:SystemModel}

We consider an FDMA system with parallel broadcast Gaussian channels in a single cell with one Base Station (BS). The base station has a down-link power budget of $\Phi$. Without loss of generality, the total power budget is assumed to be one unit, i.e. $\Phi=1$, otherwise, simply normalize all power quantities, including noise, by $\Phi$. Consider a group of $n\geq 2$ mobile end-users indexed by $i$. The channel attenuation to the user $i$ is denoted by $h_i$, and each user $i$ has a desired throughput of $b_i$ bits/s. In order to have a non-trivial problem, suppose that the BS is overloaded with the demand, i.e., the power budget is not sufficient to satisfy the demand of all end-users. 

Users bid a price, $c_i$, that they are willing to pay for one unit of data to the BS. Subsequently, the BS decides on the fraction of power she will assign to user $i$ depending on the price this user bids ($c_i$), the channel quality, and the desired throughput ($b_i$). Thus, we formulate the problem as a two-stage optimization problem:

\textit{\textbf{(1)}} In the first stage, end-user $i$ decides to bid a price $c_i$ per bit of data to maximize her expected utility function:
\begin{equation}
U_i(c_i,\vec{c}_{-i})=(v_i-c_i)T_i(c_i,\vec{c}_{-i})
\end{equation}  
where $v_i>0$ is the valuation of end-user $i$ for one bit of data, and $T_i(c_i,c_{-i})$ is the expected throughput of $i$ when she bids the price $c_i$ per bit and other users bid the vector of price $\vec{c}_{-i}$. By Shannon formula, 

\begin{eqnarray}
T_i(c_i,\vec{c}_{-i})=\log_2(1+SNR_i)
\label{eqn:shannon}\\
SNR_i=\frac{h_iP_i^*}{N_i}=q_iP_i^*,
\nonumber
\end{eqnarray}
where $q:=\frac{h_i}{N_i}$ is the channel quality.  In addition, $N_i$ and $P^*_i$ are the power of noise for $i$ and the expected dedicated downlink power of the BS to end-user $i$, respectively. Note that in the Shannon formula, without loss of generality, the bandwidth is taken to be $1$. Furthermore, in \eqref{eqn:shannon}, $SNR_i$ is the signal to noise ratio of the user $i$. 

We assumed that each end-user is assigned a distinct frequency channel and the interference between channels are assumed to be negligible. In addition, $P_i$ is dependent on the price a user bids and is determined by the BS in the second stage of the optimization problem.   

In this paper, we assume that an end-user is aware of her own channel quality and throughput she needs, and holds a belief over other users' parameters.

\textit{\textbf{(2)}} In the second stage, the BS decides on the optimum power allocation knowing the channel quality of end-users, throughput they need, and the price they bid. The goal of the BS is to maximize her revenue. The revenue of the BS depends on the price a user pays per bit, $c_i$, times the number of bits of data she uses. Thus, the revenue of the BS is $\sum_i{c_i\log_2(1+SNR_i)}$. This expression should be maximize by the BS given the prices that end-users bid in the previous stage and the channel quality of users.  Therefore, the following optimization problem is solved by the BS:

\begin{eqnarray}
\min_{\vec{P}} {-\sum_i{c_i \log_2(1+q_iP_i}}) \label{equ:optobj}\\
\qquad  s.t. \qquad \qquad \qquad \qquad \qquad   &\nonumber \\ 
\log_2(1+q_iP_i) \leq b_i \qquad \forall i \label{equ:maxthr} \label{equ:optcns1}\\
P_i\geq 0 \qquad \forall i\label{equ:optcns2}\\
\sum_i{P_i}=1 \qquad \quad &\label{equ:optcns3}
\end{eqnarray}   

The constraint \eqref{equ:optcns1} reflects the desired throughput of a user. Thus, in the optimum allocation, the BS provides at most $b_i$ bits for $i$. The constraint \eqref{equ:optcns3} resulted from the power budget of the BS and the assumption that the system is overloaded.

In the next section, we characterize the optimum power allocation strategy for the base station and the optimum bidding strategy for end-users. 

\section{Optimum Allocation and Bidding Strategies}\label{section:optimal}
We proceed to solve the problem in a reverse order. First we start with the BS problem (Section \ref{section:BTSproblem}) for allocating the power. Then we characterize the optimum bidding strategies for end-users using the optimum allocation strategy obtained by solving the BS problem (Section \ref{section:userproblem}). 

\subsection{BS Problem}\label{section:BTSproblem}
The optimization problem \eqref{equ:optobj}-\eqref{equ:optcns3} is a modified version of the conventional water-filling problem. In the conventional water-filling the objective function is merely the sum of throughputs, while in \eqref{equ:optobj} the objective function is the weighted sum of the throughput of end-users in which the weights are the price per bit that a user is willing to pay. In addition, \eqref{equ:optcns1} is an extra constraint to ensure that the allocated throughput to a user does not exceed the demand of that user. We will prove that the solution to this optimization problem is also a modified version of the conventional solution to the water-filling problem.   

Similar to the water-filling problem, the BS problem is a convex programming problem. Thus KKT conditions provides the necessary and sufficient conditions for an optimal solution, i.e. $\vec{P}^*$. The stationarity condition in KKT is as follows: 

\begin{equation}
\frac{\gamma_i^*-c_i}{\frac{1}{q_i}+P^*_i}-\lambda^*_i+\eta^*=0 \qquad \forall i \label{eqn:lag}
\end{equation}

where $\gamma_i$, $\lambda_i$, and $\eta$ are Lagrange multipliers associated with \eqref{equ:optcns1}, \eqref{equ:optcns2}, and \eqref{equ:optcns3}, respectively. The parameter $\lambda^*_i$ works similar to a slack variable. Thus $\frac{c_i-\gamma_i^*}{\frac{1}{q_i}+P^*_i}\leq \eta^*$. If $\frac{c_i-\gamma_i^*}{\frac{1}{q_i}} > \eta^*$, then $P^*_i>0$, and from complementary slackness, $\lambda_i^*=0$. Thus $\frac{c_i-\gamma_i^*}{\frac{1}{q_i}+P^*_i}=\eta^*$. On the other hand, if $\frac{c_i-\gamma_i^*}{\frac{1}{q_i}} < \eta^*$, then $P_i^*=0$. Therefore, 

\begin{equation}\label{eq:optimum}
P_i^*=\max\{0,\frac{c_i-\gamma_i^*}{\eta^*}-\frac{1}{q_i}\}=(\frac{c_i-\gamma_i^*}{\eta^*}-\frac{1}{q_i})^+
\end{equation}
which is dependent on dual parameters $\eta^*$ and $\gamma_i^*$.  Thus,
\\
\begin{theorem}
The solution to the minimization problem \eqref{equ:optobj}-\eqref{equ:optcns3} is  $P_i^*=(\frac{c_i-\gamma_i^*}{\eta^*}-\frac{1}{q_i})^+$ where $\eta^*$ and $\gamma_i^*$ are the optimum dual parameters.  \\
\end{theorem}

\textit{\textbf{Discussion:}}The solution, $P_i^*=(\frac{c_i}{\eta^*}-(\frac{\gamma_i^*}{\eta^*}+\frac{1}{q_i}))^+$, is slightly different from water-filling, in the sense that the flood level ($\frac{c_i}{\eta^*}$) is different for each user, and is dependent on $c_i$, i.e. the price the user pays. 

Either end-user $i$ receives the desired throughput, i.e., $\log_2(1+q_iP^*_i)= b_i$, or not, i.e., $\log_2(1+q_iP^*_i)< b_i$. In the latter case, $\gamma^*_i=0$, and the throughput of the user $i$ is $\log_2(1+q_i(\frac{c_i}{\eta^*}-\frac{1}{q_i})^+)$. Thus in this case, the throughput is either zero if $\frac{c_i}{\eta^*}-\frac{1}{q_i}\leq 0$, or $\log_2(q_i\frac{c_i}{\eta^*})$ if $\frac{c_i}{\eta^*}-\frac{1}{q_i}>0$. Therefore the throughput of the end-user is a non-decreasing function of $c_i$ and $q_i$:
\begin{equation}
T(c_i,\vec{c}_{-i}) =
\left\{
	\begin{array}{ll}
		\min\{\log_2(q_i\frac{c_i}{\eta^*}),b_i\}  & \mbox{if } c_i \geq \frac{\eta^*}{q_i} \\
		0 & \mbox{if }  c_i < \frac{\eta^*}{q_i}
	\end{array}
\right.
\nonumber
\end{equation}

\textit{\textbf{Discussion:}} There exists a saturation level for each user. In the saturation level, $\log_2(1+q_iP^*_i)= b_i$. In other words, the saturated level of the allocated power occurs when the desired throughput is met. 

Note that the dual parameter $\eta^*$ \eqref{eqn:lag} depends on all players' prices (including $i$) and their channel quality. This parameter is important in determining the flood level. Intuitively from \eqref{equ:optcns3} and \eqref{eq:optimum}, we can say that a higher $\eta^*$ is associated with a more congested cell, users that are willing to pay more for data, or lower quality of channels\footnote{This is not a proof. In the actual proof, one should note that $\gamma_i^*$ is changing.}. From this point on, we call $\eta^*$  the \emph{cell condition}. In section \ref{section:userproblem}, we will observe that the belief of end-users about this parameter is important in their decision making about the price they bid. 

\subsection{End-User Problem}\label{section:userproblem}
Now, consider the bidding problem for end-users. We will characterize the optimum solution(s) for the problem.

In our setting, users want to bid a price to maximize their utility function which is a decreasing function of the price they bid and an increasing function of their expected throughput: 

\begin{equation}\label{eqn:userproblem}
\begin{aligned}
\max_{c_i> 0} &U_i(c_i)=\\
&\max_{c_i> 0}  \bigg{(}(v_i-c_i)\min\{b_i,\max\{0,\log_2({q_i\frac{c_i}{{\eta}_i(c_i)}})\}\}\bigg{)}
\end{aligned}
\end{equation}
In \eqref{eqn:userproblem}, $\log_2({q_i\frac{c_i}{{\eta}_i(c_i)}})$ is the expected throughput resulted from the optimum power allocation done by the BS in the previous section, and ${\eta}_i(c_i)$ is the belief of player $i$ about the $\eta^*$ in \eqref{eqn:lag}. Note that users do not know the actual value of $\eta^*$, and the belief of a user about the cell condition is dependent on the price paid by that user. We will later discuss more about the belief function.

Bidding the price $c_i=0$ yields a throughput of zero. Therefore the optimization is on prices greater than zero, and the price that user $i$ chooses is in the set $(0,\bar{c}_i]$, where $\bar{c}_i=\inf \{c_i: \log_2({q_i\frac{c_i}{{\eta}_i(c_i)}})=b_i\}$. This happens because\ $T_i(c_i,\vec{c}_{-i})\leq b_i$ for all $c_i$. Thus for $c_i>\bar{c}_i$, $U_i(c_i)<(v_i-\bar{c}_i)b_i=U_i(\bar{c_i})$. Therefore, player $i$ does not choose a price higher than $\bar{c}_i$.  In other words, $\bar{c}_i$ is the smallest price by which a user can expect to receive the full throughput she needs. In addition, by choosing $c_i=v_i$, a user can secure the payoff of zero. Thus, the optimum price should be in the interval of $(0,\min\{\bar{c}_i,v_i\}]$. We call the interval $(0,\min\{\bar{c}_i,v_i\}]$ the feasible price interval.  

Now, consider the maximization \eqref{eqn:userproblem} on the feasible interval:
\begin{equation}
\max_{c_i\in (0,\min\{\bar{c}_i,v_i\}]}{(v_i-c_i)\log_2({q_i\frac{c_i}{{\eta}_i(c_i)}})}
\end{equation}    

The optimum bidding strategy is found by using the first order optimality condition:
\begin{equation}\label{equ:diff}
\frac{dU_i(c_i)}{dc_i}=-\log_2({q_i\frac{c_i}{{\eta}_i(c_i)}})+(v_i-c_i)(\frac{1}{c_i}-\frac{\eta'_i(c_i)}{\eta_i(c_i)})=0
\end{equation}

 Let $A_i$ denote the set of solutions to the equation \eqref{equ:diff}. The set of candidate optimum solutions, $A^*_i$, is the set of prices that are in the interval  $(0,\min\{\bar{c}_i,v_i\}]$ and satisfy the second order condition. In other words, $A^*_i=\{c: c\in A_i, c\in (0,\min\{\bar{c}_i,v_i\}], \text{ and } \frac{d^2U_i}{dc^2_i}|_{c_i=c}<0 \}$. The optimum bidding price is:

\begin{equation}
 c^*_i=\text{arg} \max \{U_i(c): c\in A^*_i \text{ or } c=\min\{\bar{c}_i,v_i\}\}
\nonumber
\end{equation}

Next, we will explain more about the belief function and provide a sample belief function for the end-users.

\textbf{\textit{The Belief Function:}} The belief function of a user is an important factor in determining the optimum bid by that user. One method by which end-users can obtain this belief function is to guess and update their belief for the cell condition by monitoring the price bid by users and channel qualities of different users for a particular interval of time. This update process is beyond the scope of this paper. 

As an another method, the belief function for the cell condition can be provided for end-users by the BS at the beginning of a time-slot. Note that the BS is aware of all channel qualities, but is not aware of the bids submitted by users before announcing the belief function. Thus, the BS can only announce an approximation of $\eta^*$ as the belief function to end-users.

An end-user expects that the belief function satisfies some properties. We call such a function a  \textit{consistent belief function}:

\begin{defn} A consistent belief function is a function which is non-decreasing with respect to $c_i$, and is increasing if the user believes that she will receive a positive power quota and the demand of the user is not met, i.e, when   $0<\log_2(1+q_i(\frac{c_i}{\eta_i(c_i)}-\frac{1}{q_i}))< b_i$. 
\end{defn}
 
The reason for consistency of such belief is discussed in the following discussion:

\textbf{\textit{Discussion:}}
Given the belief finction, $\eta_i(c_i)$, and from \eqref{equ:optcns3} and \eqref{eq:optimum}, $\sum_{j}{(\frac{c_j-\gamma_j^*}{\eta_i(c_i)}-\frac{1}{q_j})^+}=1$. Suppose that $c_j$ for $j\neq i$ is held constant and $c'_i<c''_i$. Let $\eta'^*$ and $\eta''^*$ to be the belief of the user $i$ about the optimum dual variables associated with $c'_i$ and $c''_i$, respectively. 

The claim follows by vacuity if $c'_i=\bar{c}_i$, i.e. $\log_2(1+q_i(\frac{c_i}{\eta_i(c_i)}-\frac{1}{q_i})^+)= b_i$ for the price $c'_i$. 

Trivially, if $(\frac{c_i-\gamma_i^*}{\eta_i(c_i)}-\frac{1}{q_i})^+=0$ for $c_i=c'_i$  and $c_i=c''_i$, then the belief about the optimal allocation is unchanged. Therefore $\eta'^*=\eta''^*$. 


Now consider the case that $0<\log_2(1+q_i(\frac{c_i}{\eta_i(c_i)}-\frac{1}{q_i}))< b_i$ for $c_i=c'_i$. Take $c''_i=c'_i+\epsilon$, where $\epsilon>0$ is small enough to ensure that $\frac{c_i}{\eta_i(c_i)}-\frac{1}{q_i}$ does not exceed $b_i$. Suppose that $\eta_i(c_i)$ is not a strictly increasing function of $c_i$. Therefore, $\eta'^*\geq \eta''^*$, and $P'^{belief}_i=(\frac{c'_i}{\eta'^*}-\frac{1}{q_i})^+ <(\frac{c''_i}{\eta''^*}-\frac{1}{q_i})^+=P''^{belief}_i$. In other words, the belief of the user about her power quota is increased. In addition the power quota of other users is unchanged or increased, i.e. $P'^{belief}_j\leq P''^{belief}_j$ for $j\neq i$. Thus $\sum_j {P''^{belief}_j}>1$, which is inconsistent with the belief of a user about the optimization done in the BS side.\qed

Note that the belief function is determined by the history or is provided by the BS and potentially can be any function. However, as discussed above, a function with above-mentioned properties is more consistent. In the following theorem, we present the relation between the price a user bids and her belief about the power quota she gets, when users are provided with a consistent belief. \\

\begin{theorem}

When users are provided with a consistent belief, the belief of  a user about the power allocated to her is a non-decreasing function of the price she bids for prices in the feasible interval, and it is an increasing function of the price she bids when the user believes that she will receive a positive power quota and the demand of the user is not met.
\end{theorem}

\begin{proof}
If $(\frac{c_i}{\eta_i(c_i)}-\frac{1}{q_i})^+=0$, the result follows from \eqref{equ:optcns2}. The theorem holds by vacuity if  $\log_2(1+q_i\frac{c_i}{\eta_i(c_i)}-\frac{1}{q_i})= b_i$. Now consider the case that $0<\log_2(1+q_i(\frac{c_i}{\eta_i(c_i)}-\frac{1}{q_i}))<b_i$. Suppose that $P^{belief}_i=\frac{c_i}{\eta_i(c_i)}-\frac{1}{q_i}$ is not increasing with respect to $c_i$ . The demand is consistent, therefore $\eta_i(c_i)$ is increasing function of $c_i$. Thus, $P^{belief}_j$ for $j\neq i$ decreases as $c_i$ increases. Therefore $\sum_k {P^{belief}_k}$  decreases from one. This is a contradiction with \eqref{equ:optcns3}. The result follows.\\         
\end{proof}

Next, we will present an example for the belief of end-users about the cell condition and characterize the optimum bidding strategy for this belief.

\begin{exm}\textit{Simple Belief-}
Let the set $Q_i$ denote the set of users that end-user $i$ believes that  they receive a positive throughput. Note that in this simple belief function, $Q_i$ is assumed by $i$ to be independent of $c_i$. In general, $Q_i$ would be a set that is decreasing in size with respect to $c_i$.  We introduce  parameters $C_i=\sum_{j\in Q_i}{\hat{c}_{ij}}$ and $B_i=\sum_{j\in Q_i}{\frac{1}{\hat{q}_{ij}}}$, where $\hat{c}_{ij}$ and $\hat{q}_{ij}$ are the belief of user $i$ about other user's price and channel quality. Using \eqref{equ:optcns3}, for all $i$,
$$
\frac{c_i}{\eta_i(c_i)}-\frac{1}{q_i}+\sum_{j\in Q_i}{(\frac{\hat{c}_{ij}}{\eta_i(c_i)}-\frac{1}{\hat{q}_{ij}})}=1
$$
Therefore the belief function of the user $i$ is:
\begin{equation}
\eta_i(c_i)=\frac{c_i+\sum_{Q_i}{\hat{c}_{ij}}}{1+\frac{1}{q_i}+\sum_{Q_i}{\frac{1}{\hat{q}_{ij}}}}=\frac{c_i+C_i}{1+\frac{1}{q_i}+B_i}
\end{equation} 
With this belief function, $\frac{d^2U_i(c_i)}{dc^2_i}<0$ for all feasible $c_i$. This implies that \eqref{equ:diff} has at most one solution in the feasible interval. The solution to the equation \eqref{equ:diff} is the unique optimum bidding strategy for end-user $i$ if it is  in the interval $(0,\min\{\bar{c}_i,v_i\})$; otherwise, user $i$ bids $\min\{\bar{c}_i,v_i\}$. 

A possible implementation scheme for this type of belief function is that the BS announces $C_i$ and $B_i$ at the beginning of each time-slot to the end-user $i$. Subsequently users use these value to calculate their bids. 

\end{exm}

In the next section, we present some numerical results using the \textit{simple belief function} for the cell condition. 

\section{Numerical Results} \label{section:simulation}
First we consider a cell consists of a BS and four end-users. We assume that the BS  is aware of $q_i$'s and $v_i$'s  and announces parameters $C_i$ and $B_i$ to each user $i$ as follows:
\begin{eqnarray}
C_i=\frac{1}{2}\sum_{j\neq i}{v_j}\nonumber \\
B_i=\sum_{j\neq i}{\frac{1}{q_j}}\nonumber \\
\nonumber
\end{eqnarray}

The goal of the first set of simulation is to observe the relationship between parameters of a user, i.e. $b_i$, $q_i$, and $v_i$, with the price she bids and the power allocated to her. We consider three cases. In each case, one of the previously mentioned parameters are different among users and others are the same. Note that we normalized the bandwidth and the power budget to one; therefore all other parameters should be scaled appropriately.

\begin{figure}[t]
    \centering
    \includegraphics[scale=0.35]{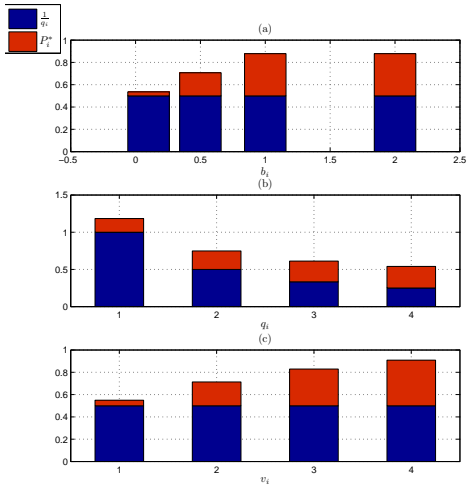}
    \caption{Generalized Water-Filling for Different Users}
    \label{fig:GWF}
\end{figure}

\textit{Case 1:} In this case, $v_i=1$ and $q_i=2$ for all end-users. The vector of the desired throughputs is $\vec{b}=[0.1,0.5,1,2]$. The price vector bid by users can be found to be $\vec{c}=[0.33,0.46,0.53,0.53]$ which is non-decreasing with respect to the throughput needed by a user. The allocated power (red part), the reverse of channel quality (blue part), and $\frac{c_i}{\eta^*}$ (sum of the blue and red parts) are plotted in Figure \ref{fig:GWF}.a. Simulation confirms  that since all users have the same quality of channel, higher price bid by a user  is associated with more power allocated to that user, and consequently a higher throughput. Thus, users with higher level of desired throughout, bid higher prices to secure higher throughputs. 

\textit{Case 2:} Now consider $v_i=1$ and $b_i=1.5$ for all users. The vector of channel qualities is $\vec{q}=[1,2,3,4]$. The vector of price bid by end-users in this case is $\vec{c}=[0.83,0.52,0.43,0.38]$ which is decreasing 
with respect to the channel quality of a user. In addition, the vector of final throghputs is $\vec{T}=[0.24,0.58,0.88,1.11]$ which is increasing by $q_i$. We can see, in Figure \ref{fig:GWF}.b, that the allocated power to different users is approximately equal for all users in spite of difference in channel qualities. The reason is because of differential pricing scheme.

\begin{figure}[t]
    \centering
    \includegraphics[scale=0.35]{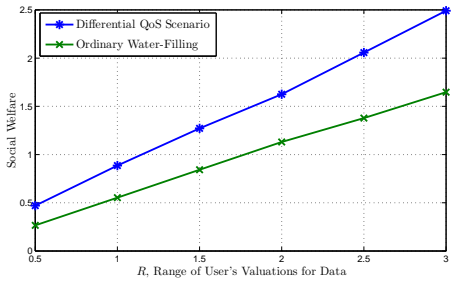}
    \caption{The comparison of social welfare in two scenarios in different ranges of user's valuation}
    \label{fig:SW}
\end{figure}

\textit{Case 3:} In this case, consider $b_i=1.5$ and $q_i=2$ for all users. The vector of valuations is $\vec{v}=[1,2,3,4]$. The vector of price bid by each user is  $\vec{c}=[0.94,1.22,1.42,1.55]$ which is increasing with respect to the valuation each user has for data. As we can see in Figure \ref{fig:GWF}.c, since users with higher valuation for data bid more generously, they obtain a higher fraction of the BS power budget. 


In another scenario, we want to investigate the effect of this differential QoS scheme on the social welfare of the cell. Here, the social welfare is defined as the sum of the payoffs of all end-users:
\begin{equation}
SW=\sum_i {U_i(c_i,\vec{c}_{-i})}=\sum_i {(v_i-c_i)T_i(c^*_i,\vec{c}_{-i})}
\nonumber
\end{equation} 

We consider a cell with one BS and 10 end-users. The belief function is as before. The valuation of users for data is distributed uniformly in the interval $[0,R]$, and the simulation is done for $R=[0.5,1,1.5,2,2.5,3]$. We consider $b_i=1.5$ and $q_i=2$ for all users. In order to alleviate the effect of randomness in results, simulation is repeated $50$ times for each value of $R$. The reported social welfare is the average of social welfare values resulted in these $50$ iterations.

Two cases are investigated: (i) Differential pricing and QoS scenario. In this case, $c_i$ is the the optimum bidding price chosen by the end-user $i$. (ii) Ordinary Water-Filling Scheme in which the BS allocates her power merely based on the channel quality of users and charges all users a flat rate of $c$. For the sake of comparison, we consider $c$ to be the average price chosen by users in scenario (i). 

Results are plotted in figure \ref{fig:SW}. Results reveal that the social welfare in the differential QoS scenario is larger than the social welfare in ordinary water-filling. Furthermore, the difference between the social welfare of two cases is increasing with the range of user's valuation for data.  

\section{Conclusion and Future Works}\label{section:conclusion}
We studied the problem of power allocation in the downlink of an FDMA system by a BS that decides on the allocation based on the price users bid, their desired throughput, and their channel qualities. The optimum power allocation and bidding strategies were characterized.  

In the future, we will extend our results to the CDMA-SINR frameworks
(where the bandwidth is fully shared). This scenario have a similar form of throughput function to that of
the SNR-FDMA framework we have studied. For the CDMA case
under overload conditions ($\sum_i P_i = 1$),
\[
T_i(\vec{c}) ~=~ \log_2(1+{\sf SINR}_i) ~=~ -\log_2(1-\tilde{q}_i P_i)
\]
where $\tilde{q}_i=h_i/(h_i+N_i)=q_i/(1+q_i)$, which is
increasing in $i$'s channel quality, $q_i$.

\end{document}